\pdfoutput=1

\documentclass[paper=a4,fontsize=10pt,DIV=12]{scrartcl}
\usepackage[english]{babel}
\usepackage[latin1]{inputenc}
\usepackage[T1]{fontenc}
\usepackage{lmodern}
\typearea[current]{last}

\usepackage{color}

\usepackage[activate,final]{microtype}
\usepackage{xspace}

\usepackage{etoolbox}

\usepackage{enumitem}
\setlist{nolistsep}
\usepackage[english=british,german=guillemets]{csquotes}
\usepackage{booktabs}\usepackage{flafter}
\usepackage{bm}

\usepackage{amsmath,amssymb,amsthm}
\usepackage{mathtools}

\usepackage{cases}
 \newtheorem{theorem}{Theorem}[section]
 \newtheorem{lemma}[theorem]{Lemma}
 \newtheorem{fact}[theorem]{Fact}
 \newtheorem{claim}[theorem]{Claim}
 \newtheorem{corollary}[theorem]{Corollary}

 \theoremstyle{definition}
 \newtheorem{definition}[theorem]{Definition}

 \usepackage{tikz}
\usetikzlibrary{calc,arrows,decorations,decorations.pathreplacing,backgrounds,shapes,fit}

\addtokomafont{sectioning}{\rmfamily}
\addtokomafont{descriptionlabel}{\rmfamily}
\makeatletter
\patchcmd{\subject@font}{\Large}{\large}{}{}
\patchcmd{\@maketitle}{\Large}{\large}{}{}
\patchcmd{\@maketitle}{\Large}{\large}{}{}
\patchcmd{\@maketitle}{\Large}{\large}{}{}
\patchcmd{\@maketitle}{\Large}{\large}{}{}
\patchcmd{\@maketitle}{\titlefont\huge}{\titlefont\LARGE}{}{}
\makeatother

\DeclarePairedDelimiter{\size}{\lvert}{\rvert}
\DeclarePairedDelimiter{\p}{\lparen}{\rparen}

\DeclarePairedDelimiter{\ang}{\langle}{\rangle}
\DeclarePairedDelimiter{\ceil}{\lceil}{\rceil}
\DeclarePairedDelimiter{\cb}{\{}{\}}
\DeclarePairedDelimiter{\br}{[}{]}

\newcommand{\eq}[2][]{\begin{equation}\label{#1}\begin{split}#2\end{split}\end{equation}}

\let\phi\varphi

\let\emptyset\varnothing

\newcommand{\csp}{\textsc{\mdseries CSP}}
\newcommand{\cnf}{\textsc{\mdseries CNF}}
\newcommand{\dcnf}[1][]{\ensuremath{#1}\textsc{-\mdseries CNF}}

\newcommand{\fol}{\textsc{\mdseries FO}}

\newcommand{\N}{\mathbb{N}}
\newcommand{\NN}{\mathbb{N}_0}

\newcommand{\W}{\textsc{\mdseries W}}

\newcommand{\pwsat}{\ensuremath{p}-\textsc{\mdseries WSat}}
\newcommand{\pcwsat}{\ensuremath{p}\textit{\mdseries -clausesize-}\textsc{\mdseries WSat}}

\usepackage{colonequals}
\newcommand{\defeq}{\colonequals}

\DeclareMathOperator{\partl}{partial}
\DeclareMathOperator{\completion}{compl}
\DeclareMathOperator{\arity}{arity}

\usepackage[final,
            pdfusetitle,
            bookmarksopen=true,
            pdfstartview=FitH,
            pdfprintscaling=None,
            unicode=true,
            ]{hyperref}

            \title{Parameterized Complexity of CSP for Infinite Constraint
              Languages
}
\author{Ruhollah Majdoddin\\% Institute für Informatik, 
  Humboldt Universität zu Berlin, Germany\\\texttt{r.majdodin@gmail.com}}
\date{}
\begin{document}
\maketitle
\begin{abstract}
 We study parameterized Constraint Satisfaction Problem  for infinite constraint
languages. The parameters that we study are weight of the satisfying assignment, number of constraints,
maximum number of occurrences of a variable in the instance, and maximum number
of occurrences of a variable in each constraint.
A dichotomy theorem is already known for finite constraint languages with
the weight parameter.
We prove some general theorems that show, as new results, that  some well-known
problems are fixed-parameter tractable and some others are in \W$\br{1}$.
 \end{abstract}

 \section{Introduction}
 A constraint language is a
domain and a set of relations over this domain. We study the parameterized complexity of Constraint Satisfaction Problem (CSP) 
for \emph{infinite} Boolean constraint languages (where each relation has a
finite arity). This is a new subject, as it
seems that the works which explicitly address
the (parameterized) complexity of \csp\ have been concerned with finite
constraint languages. 
The parameters that we study are $k$, weight of
a satisfying assignment, $k_\le$, the maximum weight of a satisfying assignment,
 the number $u$ of constraints, the maximum number $t$ of occurrences of a variable
 in the instance, and the maximum number $e$ of occurrences of a variable in a constraint. Marx
\cite{Mar05} proves a dichotomy theorem for \csp\ with parameter $k$ for finite
 constraint languages over the Boolean domain. This is extended by Marx and Bulatov \cite{BM14}
to all finite domains. Letting the constraint language to be
infinite makes the problem not just much more general, but also much more harder. Because, fore example, many proves in
\cite{Mar05} and \cite{BM14} use the fact that there is a bound on the arity of
relations in a finite constraint language, but there is no such bound for an
infinite constraint language.

 We study constraint languages that are fpt-membership checkable, that is there
is an fpt-algorithm that given
the index of a relation and a tuple, the algorithm decides whether the tuple is in
the relation (the parameter is the \textit{weight} of the tuple, that is the number of $1$s in the tuple).   

Our mathematical concepts and notation are described in detail in the next
section. Many of our results are about constraint languages $W^A$ for  some set
$A\subseteq \mathbb{N}_0$. Roughly speaking, $W^A$ has
\textit{symmetric} relations of every arity, where $A$ is the set of permitted
weights of the tuples in the relations. For an integer $c \ge 0$, $W^c$ is the
union of all $W^A$ for any $A \subseteq \br{0,c}$.

We have two groups of results. In the first group, sections $3$ and $4$, we study  \csp\ with additional
parameters besides $k$. First, we prove that   for every set $E \subseteq \NN$,
the problem \csp$\p{W^{E}}_{k, u, e}$ is   fixed-parameter tractable. Moreover, if $E$ or $\NN \setminus E$ is finite, then  \csp$\p{W^{E}}_{k, u}$ is
fixed-parameter tractable. Notice that a trivial bounded search tree method does not give an
fpt-algorithm here (See e.g. \cite[Sec. 1.3]{FG06}), because $W^{E}$ is an
infinite constraint language  and there is no bound
on the arity of the its relations. Moreover, the additional parameter $u$ is
necessary, because \csp$\p{W^{\cb{1}}}_k$, called \textsc{Weighted Exact CNF},
is \W$\br{1}$-hard \cite{DF95b}.

It then follows that for every set $E \subseteq \N$,
the problem \csp$\p{W^{E}}_{k, t, e}$ is   fixed-parameter tractable. Moreover, if $E$ or $\N \setminus E$ is finite, then  \csp$\p{W^{E}}_{k, t}$ is
fixed-parameter tractable. In the following we present some interesting examples.
\begin{itemize}
  \item \csp$\p{\W^{\mathbb{N}}}_{k,u}$ and \csp$\p{\W^{\mathbb{N}}}_{k,t}$
   are fixed-parameter tractable.

   Notice that \csp$\p{\W^{\mathbb{N}}}_{k}$ is equivalent to
   \pwsat$\p{\cnf^+}$. %with additional parameters $u$ or $t$ besides $k$.
This result is interesting, because   \pwsat$\p{\cnf^+}$ %\csp$\p{\W^{\mathbb{N}}}_{k}$
   and \pwsat$\p{\cnf}$  are both  \W$\br{2}$-complete \cite{DFR98}.
   It is noteworthy that for every $d \ge 1$, \pwsat$\p{\dcnf[d]}$ and \pwsat$\p{\dcnf[d]^-}$ are
   \W$\br{1}$-complete \cite{DFR98}, but \pwsat$\p{\dcnf[d]^+}$ % \csp$\p{\cb{\W_i^{\br{i}}|i \in       \br{d}}}_{k}$
   is fixed-parameter tractable \cite{Mar05}. Finally \pcwsat$\p{\cnf}$ is \W$\br{1}$-complete.
   \item  \csp$\p{W^{\cb{1}}}_{k, u}$ and \csp$\p{W^{\cb{1}}}_{k, t}$ are
    fixed-parameter tractable.

    This is the problem \textsc{Weighted Exact CNF}  with the additional parameter $u$ or $t$.
    \item \csp$\p{W^{\text{odd}}}_{k, t\le 3}$ and \csp$\p{
  W^{\text{even}}}_{k\le, u}$ are fixed-parameter tractable. 

It is proved in \cite{DFVW99} that \csp$\p{
  W^{\text{odd}} \cup  W^{\text{even}}}_{k\le}$, \csp$\p{
  W^{\text{odd}}}_{k\le}$ and even \csp$\p{
  W^{\text{even}}}_{k}$ are \W$\br{1}$-hard. Arvind et al. \cite{AKKT14} proved
that the hardness holds also with the additional parameter $t$, and even if $t$ is bounded
to $t\le 3$, that is \csp$\p{W^{\text{odd}} \cup  W^{\text{even}}}_{k\le, t\le
  3}$ and  \csp$\p{
  W^{\text{even}}}_{k, t\le 3}$ are \W$\br{1}$-hard.  Notice that both papers study these
problems in the setting of linear equations $Ax=b$ over $\mathbb{F}_2$. But, to
the best of our knowledge, the complexity of \csp$\p{
  W^{\text{odd}}}_{k, t\le 3}$ has been left open, and by our result above, it is
fixed-parameter tractable. We find this somewhat surprising, as it is the only known case of parameterized \csp\ over
\textit{parity} constraint languages that introducing the additional
parameter $t$ reduces the complexity of the problem. An important  open problem
is whether
\csp$\p{
  W^{\text{even}}}_{k\le}$, called \textsc{Even Set}, is \W$\br{1}$-hard \cite{CFJK14}. It is even not known whether 
\csp$\p{
  W^{\text{even}}}_{k\le, t}$ is \W$\br{1}$-hard. Still, we  prove that  \csp$\p{
  W^{\text{even}}}_{k\le, u}$ is fixed-parameter tractable.
  \end{itemize}
Also, it has  been left open whether \csp$\p{
  W^{\text{odd}} \cup  W^{\text{even}}}_{k, t}$ is in
\W$\br{1}$. Our second result answers this positively: 
For every (possibly infinite) constraint language $\Gamma$, it holds
\csp$\p{\Gamma}_{k, t}$ $\in$ \W$\br{1}$. 

Our second group of results, presented in sections $5$ and $6$, is about
containment in \W$\br{1}$. Downey and Fellows showed that that  \textsc{Weighted Exact CNF}, in our setting
\csp$\p{W^{\cb{1}}}_{k}$, reduces to $p$-\textsc{Perfect Code} and vice versa,
and proved that the problem is \W$\br{1}$-hard by means of a reduction from
$p$-\textsc{Independent-Set} \cite{DF95b, DF99}. They conjectured that the problem could be of difficulty intermediate between \W$\br{1}$
and \W$\br{2}$, and thus not \W$\br{1}$-complete \cite[pp. 277, 458, 487]{DF99}.
%(As noted above, we prove that with extra parameter $t$, the problem is fixed-parameter tractable). 
Surprisingly, Cesati  proved that \csp$\p{W^{\cb{1}}}_{k}\in \W\br{1}$\cite{Ces02, Ces03}. He uses reductions
to \textsc{Short NonDeterministic Turing Machine Acceptance}. In this problem the Turing
Machine can have fpt-size state space and alphabet, but the parameter is the runtime
of the machine.

A natural question is whether this result can be generalized. First, by a somewhat
technical proof, which is an adaption of  the proof of \cite{Ces02}, we show that for every $d \ge 0$,
\csp$ \p{CW^{\br{d}}}_{k}$ $\in$ \W$\br{1}$. Notice that $W^{\br{d}}
\subset CW^{\br{d}}$. 
We generalize this, by still another involved proof, as follows: If for a (possibly
infinite) constraint language  $\Gamma$, there be  an integer
$d \ge 1$ such that for every relation $R \in \Gamma$, size of each set $T \in
R$ is at most $d$, then \csp$\p{\Gamma}_{k} $ $\in$ \W$\br{1}$.
This implies our ultimate generalization: %for \csp$\p{W^{\cb{1}}}_{k}$:
 For every $d \ge 0$,   \csp$ \p{W^d}_{k} \in  \W\br{1}$. 

 \section{Preliminaries}

We denote the set of positive integers by $\N$ and the set of
nonnegative integers by $\NN$.
For integers $0 \le a \le b$, we use the notation $\br{a, b} \defeq \cb{a, a+1,
  \ldots, b}$. For an integer $a > 0$, we write $\br{a} \defeq \br{1, \ldots, a}$, and
we denote $\br{0}\defeq \emptyset$. We encode integers in binary.

We use the abbreviations
\eq{
&  \text{even} \defeq \cb{i| i = 2j, \, j \in \NN },\\
&  \text{odd} \defeq \cb{i| i = 2j+1, \, j \in \NN}.
  }
For a set $X$, we denote the powerset of $X$ by $2^X$. For sets $A$ and $B$ and a function $f: A
\rightarrow B$,  the \textit{image of a set $D\subseteq A$ under $f$} is
defined as $f\p{D}
\defeq \cb{y \in B| \exists x \in A,  y = f\p{x}}$, and the
\textit{preimage of a set $E\subseteq B$ under $f$} is defined as $f^{-1}\p{E}\defeq \cb{x \in A|
  f\p{x}\in E}$.

We use the basic definitions of parameterized complexity theory in
\cite{FG06}, including the following.
  Let $\Sigma$ be a nonempty finite alphabet. We refer to sets  $Q\subseteq\Sigma^*$  of strings over
$\Sigma$ as \textit{classical} decision problem. A \textit{parameterization} of
$\Sigma^*$ is a mapping $\kappa : \Sigma^* \rightarrow \N$ that is polynomial
time computable. A \textit{parameterized problem} (over $\Sigma$) is a pair
$\p{Q, \kappa}$ consisting of a set $Q \subseteq \Sigma^*$ of strings over
$\Sigma$ and a parameterization $\kappa$ of $\Sigma^*$. 

Let $D$ be a set. A relation $R$ of arity $\arity\p{R} \ge 1$ over domain $D$ is a
subset:  $R \subseteq D^{\arity\p{R}}$. In this paper we consider only relations
of finite arity.
 A set $\Gamma$ of relations over $D$,  $\Gamma \defeq\cb{R_i}_{i \in \N}$, is
 called a \textit{constraint language over domain $D$}. For a constraint $R_i
 \in \Gamma$, $i$ is called the \textit{index} of $R_i$ (in $\Gamma$).

In this paper, other than in the context of structures and first order logic
(Section \ref{section-fpt}), we
always have \textit{Boolean domains}. That is
$D=\cb{0,1}$, and we identify a
tuple in $D^r$ with a subset of $\br{r}$, which is the set of $1$ positions of
the tuple. This results in the  definition of a \textit{Boolean
  relation} as a set of subsets: $R \subseteq 2^{\arity\p{R}}$. This is the
definition that we use in this paper. A Boolean constraint language is  defined likewise. 

A constraint language $\Gamma\defeq\cb{R_i}_{i \in \N}$ is called \textit{membership
checkable}  if there is a Turing machine $B$ (membership checker) such that for any relation $R_i\in
\Gamma$ of arity say $q$, and for any $T \subseteq \br{q}$,
the machine $B$ with input $\p{i, T}$
%in binary format
outputs $1$ if $T\in R_i$, and $0$ otherwise. We say $\Gamma$ is \textit{fpt membership checkable} if there is a
computable function $f_B:\NN \rightarrow \cb{0,1}$ and an integer $c$ 
such that $B$ halts in at most 
$f_B\p{\size{T}} \p{\log i}^c$ steps. We assume that all the  constraint languages
in this paper are  fpt membership
checkable. This is arguably a mild condition and includes many interesting
natural problems. 

  For a set  $A \subseteq \NN$, we define the relation 
 $W^A_m$ of arity $m$ as
\eq{
  W^A_m\defeq \cb{T | T\subseteq \br{m}, \, \size{T} \in A},
  }
and the constraint language $W^A = \cb{W^A_i}_{i \in \N}$. The constraint
languages $W^{odd}$ and $W^{even}$  are defined
accordingly. For an integer $c \ge 0$, define the constraint language $W^c$ as
\eq {
W^c \defeq  \bigcup_{A \subseteq \br{0, c}} W^{A}.
  }
  Define the constraint language $W$ as
\eq {
W \defeq  \bigcup_{A \subseteq \NN} W^{A}.
  }

One interesting example is
$W^{\br{m}}_m$, which is equivalent to disjunction. So  $W^{\N}$ is the
constrain language  of disjunctions of any arity.  The other interesting example is
$W^{\cb{1}}$, which can be used to define the \textsc{Weighted Exact  CNF} problem.

For  a set  $A \subseteq \NN$ and  integers $d, m \ge 0$, we define
the relation \textit{conditional weight}, $ CW^A_{d,m}$, of arity $d+m$, as
\eq{
CW^A_{d,m} \defeq \cb{T \subseteq \br{d+m}\;|\; \br{d} \subseteq T
  \Rightarrow  \size{\br{d+1, d+m} \cap T} \in A},
}
and the constraint language
\eq{
  CW^A \defeq \cb{CW^A_{d, m}}_{d, m \ge 0}.
  }
% For integer $c \ge 0$, define the constraint language $CW^c$ as
% \eq {
% CW^c \defeq  \bigcup_{A \subseteq \br{0, c}} CW^{A}.
% }
A degenerate case is $d=0$, where
\eq{
  CW^A_{0,m} \defeq W^A_m.
}
Another degenerate case is $0 \notin A$ and $m=0$, then $CW^A_{d, 0}$ is
equivalent to Nand. A Horn clause  can be expressed as $CW^{1}_{d, 1}$.
  
For a constraint language  $\Gamma=\cb{R_i}_{i \in \N}$ and a set of
\textit{variable symbols} (or simply \textit{variables}) $V$, a
\textit{constraint} $\psi$  is a pair $\psi \defeq \p{m, \p{v_1, \ldots,
    v_r}}$, where $m \in \N$, $r \defeq
\arity\p{R_m}$ and $\p{v_1, \ldots, v_r}
\in V^r$. We denote $\psi$ either by $R_m\ang{v_1, \ldots, v_r}$, or by
$R_m\ang{\omega}$, for the mapping $\omega:\br{r}
\rightarrow V$, where for every $i \in \br{r}$, $\omega\p{i} \defeq
v_i$.

An \textit{assignment} $D$ of $V$ is a subset of $V$. We say that $D$ satisfies
$\psi$ if $T \in R_m$, where $
  T \defeq \cb{i | i \in \br{r}, v_i \in D}$.

We denote \csp\ for a constraint language $\Gamma$ by \csp$\p{\Gamma}$. An
instance $I$ of \csp$\p{\Gamma}$ is a pair $I\defeq\p{V, C}$,
where $V$ is the set of variables and $C$ is a sequence of constraints for $\Gamma$. We write $\size{C}$ to denote the length of $C$, and write $C\p{i}$ to
denote the $i$th constraint in $C$. An
assignment $D$ of $V$ satisfies $I$, if $D$ satisfies every constraint in
$C$.

Parameterized \csp\ has exactly one of the parameters $k$ and $k_\le$,  and possibly any of the parameters $t, u$ and $e$.

Parameter $k$ is called \textit{weight}, and an instance of \csp\ with this parameter, \csp$\p{\Gamma}_k$, is a pair $I
\defeq \p{V, C}$, where  $C$ is defined as in \csp$\p{\Gamma}$, but additionally
the first constraint is of form $C\p{1}\defeq W_{\size{V}}^{\cb{k_0}}\ang{v_1, \ldots, v_m}$, where
$V\defeq \cb{v_1, \ldots, v_m}$. We define $k\p{I}\defeq
k_0$. This means that if some $D\subseteq V$ is a satisfying assignment, then $\size{D} = k_0$.  

Likewise, parameter $k_\le$ is called \textit{weight-less-than}, and an instance
of \csp\ with this parameter, \csp$\p{\Gamma}_{k_\le}$, is a pair $I
\defeq \p{V, C}$, where  $C$ is defined as in \csp$\p{\Gamma}$, but additionally
the first constraint is of form $C\p{1}\defeq W_{\size{V}}^{\br{k_0}} \ang{v_1, \ldots, v_m}$, where
$V\defeq \cb{v_1, \ldots, v_m}$. We define $k_\le\p{I}\defeq
k_0$. This means that if some $D\subseteq V$ is a satisfying assignment, then $\size{D} \le k_0$.  

It was possible, and is common in the literature, to give the (integer) value of
the parameters $k$ and $k_{\le}$ as  part of the input. But our definition of the
parameterized problem is arguably more homogeneous, as it demonstrates that these weight conditions are
indeed just another constraint.

For an instance $I\defeq\p{V, C}$ of  \csp$\p{\Gamma}_{k, u, t, e}$,  $u$ is the
number of constraints,
\eq{
  u\p{I} \defeq\size{C},
}
 $t$ is the maximum number of occurrences of any variable  in the whole instance,
\eq{
t\p{I} \defeq \max_{v \in V} \sum_{\substack{i \in \br{\size{C}}\\C\p{i} = R\ang{\omega}}} \size{\omega^{-1}\p{\cb{v}}},
  }
  and  $e$  is the maximum number of occurrences of any variable  in any constraint,  
  \eq{
e\p{I} \defeq \max_{v \in V} \max_{\substack{i \in \br{\size{C}}\\C\p{i} = R\ang{\omega}}} \size{\omega^{-1}\p{\cb{v}}}.
}

We write $\csp\p{\Gamma_1}_t\p{\Gamma_2}_{k_\le}$ to denote  parameterized \csp\ for
the constraint language $\Gamma_1 \cup \Gamma_2$, where parameter $k$ is defined as
usual, but the  parameter $t$ applies only to the constraints for $\Gamma_1$. More
formally, for an instance $I\defeq\p{V, C}$
\eq{
t\p{I} \defeq \max_{v \in V} \sum_{\substack{i \in \br{\size{C}}\\C\p{i} = R\ang{\omega}\\R \not\in \Gamma_2}} \size{\omega^{-1}\p{v}}.
}

Sometimes the problems are restricted to those instances with a fixed or bounded
parameter value. We denote this in the parameter list. For
example, for each instance $I$ of problem \csp$\p{\Gamma}_{k, t\le 3}$, it holds
$t\p{I} \le 3$. 

\begin{fact}
  \label{k-le-to-k}
  For every constraint languages $\Gamma$, $\csp\p{\Gamma}_{k_{\le}}
  \le^{\text{fpt}} \csp\p{\Gamma}_{k}$.
\end{fact}

We use the  notation of \cite{FG06} for propositional logic , including the following. We denote by \cnf\ the class of all propositional formulas in \textit{conjunctive
  normal form}.  $\cnf^+$ denotes the subclass of \cnf\ that there is no  negation
symbol in a formula, and $\cnf^-$ denotes the subclass of \cnf\ that there is a
negation symbol in  front of every variable in a formula.

For any class
          $A$ of propositional formulas, the \textit{parameterized weighted satisfiability
            problem for}  $A$ is defined as follows:
          \vspace{3mm}

\noindent
\begin{tabular}{ l l }
\textsc{\pwsat$\p{A}$}  \\
  \textit{Instance:} &$\alpha \in A$ and $k \in \mathbb{N}$.  \\
\textit{Parameter:} &$k$.\\
\textit{Question:} &Decide whether $\alpha$ is $k$-satisfiable.\\
\end{tabular}

          \vspace{3mm}
For class \cnf\, we also define 
          \vspace{3mm}

\noindent
\begin{tabular}{l}
\textsc{\pcwsat$\p{\cnf}$}  \\
\begin{tabular}{@{}l l@{}}
\textit{Instance:} &$\alpha \in $ \cnf\ and $k \in \mathbb{N}$.  \\
\textit{Parameter:} &$k+d$, where $d$ is the maximum number of literals in a
clause of $\alpha$.\\
\textit{Question:} &Decide whether $\alpha$ is $k$-satisfiable.\\
\end{tabular}
\end{tabular}
\vspace{3mm}

We denote the size of input by $n$.
\section{Number of Occurrences of Variables - in FPT}
\label{section-fpt}
We use the
notation for first-order logic from \cite{FG06}. The class of all
first order formulas is denoted by \fol. 

A \textit{(relational) vocabulary} $\tau$ is a set of relation symbols. Each
relation symbol $R$ has an \textit{arity} $\arity\p{R} \ge 1$. A
\textit{structure} $\cal A$ of  vocabulary $\tau$, or \textit{$\tau$-structure} (or simply
 \textit{structure}), consists of a set $A$ called the \textit{universe} and
and interpreteation $R^{\cal A} \subseteq A^{\arity\p{R}}$ of each  relation symbole $R
\in \tau$. We write $R^{\cal A}\bar a$ to denote that the tuple $\bar a \in
A^{\arity\p{R}}$ belongs to the relation $R^{\cal{A}}$.

The \textit{parameterized model-checking problem} for a class $\Phi$ of formulas
is defined as follows.\vspace{3mm}

\noindent
\begin{tabular}{ l l }
\textsc{$p$-MC$\p{\Phi}$}  \\
  \textit{Instance:} &A structure $\cal A$ and a formula $\varphi \in \Phi$.  \\
\textit{Parameter:} &$\size{\varphi}$.\\
\textit{Question:} &Decide whether $\phi\p{\cal A} \ne \emptyset$.\\
\end{tabular}

The restriction of $p$-MC$\p{\Phi}$ to input structures from a class $D$ of
structures is denoted by $p$-MC$\p{D, \Phi}$.

Let $A$ be a $\tau$-structure. The \textit{Gaifman graph} (or \textit{primal
  graph}) of a $\tau$-structure $\cal A$ is the graph $G\p{\cal{A}}$$\defeq \p{V, E}$, where
$V \defeq A$ and
\eq{
  E \defeq \cb{
    \cb{a, b} |     &a,b \in A, a \ne b, \text{ there exists an } R\in \tau \text{ and a tuple }\\
    &\p{a_1, \ldots, a_r} \in R^{\cal A}, \text{ where  } r \defeq \arity\p{R}, \text{such
      that }\\
    &a,b \in \cb{a_1, \ldots, a_r}
    }
}
The \textit{degree} of a structure $\cal A$ is the degree of its Gaifman graph
$G\p{\cal A}$. 

The following corollary follows from \cite{FG01}.
% by observing that the class of all structures of degree at most $d$ has
% effectively bounded local tree width. 
\begin{theorem}[see {\cite[Corollary 12.23]{FG06}}]
  \label{structure-degree-fpt}
     Let $d \in \mathbb{N}$. The parameterized model-checking problem for
     first-order logic on the class of structures of degree at most $d$ is
     fixed-parameter tractable. 
     \end{theorem}
     \begin{theorem}
       \label{k-u-fpt}
       Let $E \subseteq \NN$. Then \csp$\p{W^{E}}_{k, u, e}$ is 
           fixed-parameter tractable. Moreover, if $E$ or $\NN \setminus E$ is finite, then  \csp$\p{W^{E}}_{k, u}$ is
          fixed-parameter tractable.
%Important: Let W be class of all weight relations. then if
%the maximum weight (or maximum of complement of weight) used in the instance is taken as parameter, w,
%\csp(W)_{k,u,w} is in FPT.          
\end{theorem}
  \begin{proof}
To each  instance $I=\p{V, C}$ corresponds a structure $\cal A$ and a formula
$\phi \in \fol$, described below, such that $\phi\p{\cal{A}} \ne \emptyset$ if and only if $I
$ has a solution.

Let $h\defeq \min \cb{e\p{I}, \max E, \max \p{\NN \setminus E}}$. Let the universe of $\cal A$ be $A\defeq
V$, and $\tau \defeq \cb{R_{i, j} | i \in \br{2, \size{C}}, j \in \br{0, h}}$,
where each  $R_{i,j}$
is a unary relation as follows. Let the $C\p{i}\defeq W_d^E\ang{\sigma}$. Then
$\forall v \in V \quad R_{i,j}^{\cal A} v$ if and only if $\size{\sigma^{-1}\p{v}}=j$. 
Because the vocabulary $\tau$ has only unary relations, the degree of structure
$\cal A$ is $0$.
Let $D$ be the class of all structures constructed this way. It follows from
Theorem \ref{structure-degree-fpt}  that $p$-MC$\p{D, \textsc{FO}} $ is
fixed-parameter tractable.

Let $k_0 \defeq k\p{I}$. Let formulas $\phi_1$ and $\phi_2$ be as follows
\eq{
  \phi_1 \defeq \exists x_1 \ldots \exists x_{k_0}  \p[\Big]{
    &\bigwedge_{\substack{i, j \in \br{{k_0}}\\ i\ne j}}
    x_i \ne x_j \\
 \land &\bigwedge_{i \in \br{2, \size{C}}}\;
  \bigvee_{\substack{c_1, \ldots, c_{k_0} \in \br{0, h}\\ \sum c_i \in E}}
  \;\bigwedge_{j \in \br{k_0}} R^{\cal{A}}_{i, c_j} x_j  
},
}
\eq{
  \phi_2 \defeq \exists x_1 \ldots \exists x_{k_0}  \p[\Big]{
    &\bigwedge_{\substack{i, j \in \br{{k_0}}\\ i\ne j}}
    x_i \ne x_j \\
 \land &\bigwedge_{i \in \br{2, \size{C}}}\;
  \bigwedge_{\substack{c_1, \ldots, c_{k_0} \in \br{0, h}\\ \sum c_i \not\in E}}
  \;\lnot \bigwedge_{j \in \br{k_0}} R^{\cal{A}}_{i, c_j} x_j
}.
}
If $E$ is finite, set $\phi \defeq \phi_1$. If $\NN \setminus E$ is finite, set
$\phi \defeq \phi_2$. If $e\p{I}$ is finite, then $\phi$ can  be set as $\phi_1$ or $\phi_2$.
It is straightforward to see that $\size{\phi}$ depends only on $k_0$, $h$ and $u\p{I}$.
This completes the proof.
%     Rreduce it to     Weighted Independent set Problem which is in FPT (by DLMR12,
% Thm. 1, 2, ) or reduce it to Bipartite Graph, degree of left is $\le \alpha$,
% cover the right by a set of link such that each on right has exactly one
% selected on link, then use and Flum-Grohe Thm 12.22. 
    \end{proof}
    Notice that by Fact. \ref{k-le-to-k}, the above result still holds if paramater $k$ is
    replaced by $k_{\le}$.
    \begin {fact}\label{parity-without-t-u}
\csp$\p{W^{\text{odd} }\cup W^{\text{even}}}_{k,u, e}$ and \csp$\p{W^{\text{odd} }\cup W^{\text{even}}}_{k,t, e}$
 are fpt-reducible   to \csp$\p{W^{\text{odd} }\cup
   W^{\text{even}}}_{k,u}$ and \csp$\p{W^{\text{odd} }\cup
   W^{\text{even}}}_{k,t}$, respectively. The reduction can be computed in polynomial time.
\end{fact}
By the above fact, we get the following corollary.
    
    \begin{corollary}
      \csp$\p{\W^{\N}}_{k,u}$, \csp$\p{W^{\cb{1}}}_{k, u}$ and
      \csp$\p{W^{\text{odd} }\cup W^{\text{even}}}_{k,u}$ are fixed-parameter tractable. 
      \end{corollary}
    Observing in Theorem \ref{k-u-fpt} that if $0 \not \in E$ and  $u\p{I} > t\p{I} k\p{I}+1$,
    then the instance has no satisfying assignment, and with Fact \ref{parity-without-t-u}, we get the following
    corollary about when $k$ and $t$ are the parameters. 
    \begin{corollary}
Let $E \subseteq \mathbb{N}$ (so $0 \not \in E$). Then \csp$\p{W^E}_{k, t, e}
$ is fixed-parameter tractable. Moreover, if $E$ or $\N \setminus E$ is finite, then  \csp$\p{W^{E}}_{k, t}$ is
          fixed-parameter tractable. Especially  \csp$\p{\W^{\mathbb{N}}}_{k,t}$,
\csp$\p{W^{\cb{1}}}_{k, t}$,
      \csp$\p{W^{\text{odd}}}_{k,t}$ are fixed-parameter tractable. 
      \end{corollary}

\section{Number of Occurrences of Variables - in W[1]}
The following problem is  $W\br{1}$-complete \cite{HW94, CCDF97}.\vspace{3mm}

\begin{tabular}{@{}l}
\textsc{Short NonDeterministic Turing Machine Acceptance}\\
\begin{tabular}{ @{}l p{10cm}}
\textit{Instance:} &A single-tape, single-head nondeterministic
Turing machine $M$; a word $x$ over the alphabet of $M$;
a positive integer $l\in\mathbb{N}$.\\
\textit{Parameter:} &$l$.\\
\textit{Question:} &Is there a computation of $M$ on input $x$ that
reaches a final accepting state in at most $l$ steps?\\
\end{tabular}
\end{tabular}

\begin{theorem}
  \label{appearance}
Let $\Gamma$ be  a (possibly infinite) constraint language. Then 
  \csp$\p{\Gamma}_{k, t}$ $\in$ \W$\br{1}$. 
\end{theorem}
\begin{proof}
  Let $\Gamma\defeq\cb{R_i}_{i \in \N}$ and the Turing machine $B$ be the fpt membership checker
  of $\Gamma$.  We present an fpt-reduction that maps  any given  instance
$I=\p{V, C}$ (with the parameter values $k\p{I}=k_0$ and $t\p{I}=t_0$) to an instance $J$ of \textsc{Short
  NonDeterministic Turing Machine Acceptance}.
Let $m\defeq \size{C}$, and the  constraints from $\Gamma$ in $C$ be $R_{h_i}\ang{\omega_i}$
for $i \in \br{2, m}$.

 For a variable $v \in V$, define the set
$E_v$ as
\eq{
E_v \defeq  \cb{i | i \in \br{m}, \, \omega_i^{-1} \p{\cb{v}} \ne \emptyset}.
}
Notice that $\size{E_v} \le t_0$. Let $D$ be the set
\eq{
  D \defeq
  \cb{i | i \in \br{m}, \, \emptyset \not \in R_{h_i}}.
}
If $\size{D} > k_0\cdot t_0$, then the machine $M$ of
instance $J$ simply rejects. Otherwise $M$, \emph{knowing} the set $D$ and also the
sets $E_v$ for any
variable $v \in V$, starts with an empty
tape and runs in three steps.\\
  \textit{Step 1.
%  :  Guess an assignment of size $k$.
  }$M$ chooses nondeterministically the set of variables
  $A\subseteq V$ of size $\size{A} \le  k_0$. \\
  \textit{Step 2.} For each $i \in \br{2, m}$ % constraint
                                % $R_{h_i}\ang{\omega_i}$ of the   instance,
  such that $i \in E_v$ for some $v \in A$,
  $M$ computes the set $T \defeq \cb{j |\omega_i\p{j}
    \in A}$.
  $M$ \textit{runs } $B\p{h_i, T}$, and if the output is $0$, then $M$ rejects.\\
  \textit{Step 3.} If $D \not \subseteq \bigcup_{v \in A} E_v$, then $M$
  rejects. Otherwise $M$ accepts.

  We should now show that machine $M$ can indeed run machine $B$ efficiently.
  This assures that parameter $l$ of instance $J$, that is runtime of $M$, is a fixed function of $k_0$ and
  $t_0$.  We   will use the following basic theorems. \vspace{2mm}
  
  Linear Speedup Theorem (see {\cite{Pap95}[Theorem 2.2]}) - 
    \label{linear-speedup}
\textit{    Let $L \in TIME\p{g\p{n}}$. Then, for any $\epsilon > 0$, $L \in
  TIME\p{g'\p{n}}$, where $g'\p{n}=\epsilon g\p{n} + n + 2$.}
\vspace{2mm}
    
  In the proof of Linear Speedup Theorem, if $G$ is the machine deciding $L$
 with runtime $g\p{n}$ and $G$ has $d$ tapes and alphabet $\Sigma$, then a
  machine $G'$ with runtime $g'\p{n}$ is constructed, such that $G'$ has $d$ tapes if $d>2$ and $2$ tapes if $d
  \le 2$, and its alphabet  is  $\Sigma'= \Sigma \cup
  \Sigma^{\ceil{\frac{6}{\epsilon}}}$. The cause of the $n + 2$ part of $g'\p{n}$ is  that
  $G'$ should scan the whole input and translate it in its own alphabet. 
  \vspace{2mm}
  
$1$-tape Simulation Theorem  (see {\cite{Pap95}[Theorem 2.2]}) - 
    \label{one-tape-tm}
    Given any $t$-tape Turing machine $H$ operating within time $f\p{n}$, we
    can construct a $1$-tape Turing machine $H'$ operating within time
    $O\p{f\p{n}^2}$     and such that, for any input $x$, $H\p{x}=H'\p{x}$. 
    \vspace{2mm}
    
In the proof of the above theorem, if $H$ has alphabet $\Sigma$, then the  size
of the alphabet of $H'$ is $2\size{\Sigma} + 2$.

Remember that runtime of $B$ on input $\p{i, T}$ is, by definition,
$f_B\p{\size{T}} \p{\log i}^c$ for some fixed $c$. 
Now let $\Sigma$ be the alphabet of $B$. Then  by applying the  Linear Speedup Theorem for $\epsilon = {1}/{\p{\log i}^c}$, there is a machine
  $B'$ that has two tapes and an alphabet $\Sigma'$ of size $\size{\Sigma'} = \size{\Sigma}+
  \size{\Sigma}^{\p{\log i}^c}$ and has runtime
  $f_B\p{\size{T}} + O\p{\log i + \size{T} + 2}$, such that $B'$ and $B$ have the same
  output on all inputs. Therefore, by $1$-tape Simulation Theorem above,  machine $M$
  having one tape   and an alphabet of size at least $2 \Sigma' + 2$ can \textit{simulate} $B'$ in time
  $O\p{\p{f_B\p{\size{T}}}^2}$. The $O\p{\log i + \size{T} + 2}$ part of the runtime can be
  avoided because $B'$ \textit{simulated} by $M$ does not have to translate the
  input to its alphabet.
\end{proof}
%{\color{blue} $\p{W}_{k, t}$}

% Theorem \ref{appearance} cannot be proved by treewidth, because bipartite graphs
% with bounded degree generally do not have bounded degree (Lozin, Rautenbach 2004). 

\begin{corollary}
  \label{parity-appearance-w1}
  % If a constraint language $\Gamma$ is uniformly-symmetric, then
  % \csp$\p{\Gamma}_{k, t}$ $\in$ \W$\br{1}$.
   \csp$\p{W^{\text{odd}} \cup W^{\text{even}}}_{k, t} \in $ \W$\br{1}$. 
 \end{corollary}

\section{CW Formulas}

% \begin{corollary}\label{in-w1-beta}
% For all $\beta \ge 1$,
%  \csp$\p{\bigcup_{\alpha \le\beta} CW^{\br{\alpha}}}_{k}$ is in  \W$\br{1}$.
% \end{corollary}

\begin{theorem}\label{reduction}
Let $b\ge 0$. Then
\csp$ \p{CW^{\br{b}}}_{k}$ is in \W$\br{1}$.
 %  Moreover, there is a constant $c$, such that for all $\lambda \ge 0$, the
 %  reduction for  an instance $I \in$
 %  $\bigcup_{e \ge 0}$ \csp $\p{\cb{CW_{d, m}^{\br{e}}}_{0
 %      \le d; m\le \lambda}}_{k}$ can be computed in time    
 % $O\p{ n^c \sum_{i=0}^{e+1}\binom{\lambda}{i}}$.
\end{theorem}
\begin{proof}
  We present an fpt-reduction from \csp $\p{CW^{\br{b}}}_{k}$ to \textsc{Short
    NonDeerministic Turing Machine Acceptance}.  
Any given instance $I\defeq \p{V,C}$ % t=t_0$ and $l =l_0$,
with parameter $k\p{I}\defeq k_0$ is mapped to an instance
  $J$ which is a Turing machine $M$ with parameter $f\p{k_0}$, $f$ to be fixed
  later.  
  
  Let $B, G \subseteq V$. Define $\Delta_{B, G}$ as the set
of integers $i \in \br{2, \size{C}}$ such that for the constraint
$C\p{i}\defeq CW_{d, m}^{\br{b}}\ang{\omega}$, it holds that $\omega\p{\br{d}} =B$ and $G
\subseteq \omega \p{\br{d+1, d+m}}$.
Define $\Lambda_{B, G}$ as
\eq{
  \Lambda_{B, G} \defeq
\max_{\substack{C\p{i}\defeq CW_{d, m}^{\br{b}}\ang{\omega} \\i \in \Delta_{B, G}}} \size{\omega^{-1}\p{G} \cap \br{d+1, d+m}}. 
}
\begin{fact}\label{claim2}
  Let $B, A \subseteq V$. Then $ \size{\Delta_{B, \emptyset}} = \size{\bigcup_{v \in A
    % \setminus \rho
   } \Delta_{B,
        \cb{v}}   }$ if and only if for
every constraint $CW_{d, m}^{\br{b}}\ang{\omega}$ in $C$ such that $\omega\p{\br{d}}
        =B$, it holds that $A \cap  \omega \p{\br{d+1, d+m}} \ne \emptyset$.
\end{fact}
%\begin{proof}

%  , and therefore $P$ satisfies the constraint. 
%   If $B$ does not satisfy a constraint $ \bigwedge_{v \in G} v
%   \Rightarrow \quad \bigvee_{v \in H}^{\le b}  v$, then
% $G \subseteq B$. Now the equality implies that $A \cap H\ne \emptyset$, and therefore $P$
% satisfies the constraint.
%\end{proof}
To take advantage of the above fact,  $\size{\bigcup_{v \in A
    % \setminus \rho
   } \Delta_{B,
        \cb{v}}   }$ should be computed. The following claim
shows how to do it  efficiently.
\begin{claim}\label{claim1}
For every $B, A \subseteq V$, for any  constraint      $CW_{d,
  m}^{\br{b}}\ang{\omega}$ in $C$ such that $\omega\p{\br{d}}
=B$ implies $\size{\omega \p{\br{d+1, d+m}} \cap A} \le b$,   it holds that
      
 \eq{\label{inclusion-exclusion}
\size[\Big]{\bigcup_{v \in A
    % \setminus \rho
   } \Delta_{B,
        \cb{v}}} = \sum_{\substack{ G \subseteq A \\ 0 < \size{G} \le b}}
(-1)^{\size{G}-1} \; \size*{\Delta_{B,G} }.
}
\end{claim}
   \begin{proof} 
\eq{\label{computation-formula}
\size[\Big]{\bigcup_{v \in A
       %  \setminus \rho
       } \Delta_{B,
       \cb{v}}} &=
 \sum_{\substack{ G \subseteq A
%     \setminus      \rho
     \\ 0 < \size{G} }}
 (-1)^{\size{G}-1}\size[\Big]{\bigcap_{v \in G} \Delta_{B,
     \cb{v}}}\quad \quad\quad \text{(inclusion-exclusion principle) }\\
 &=\sum_{\substack{ G \subseteq A
     %\setminus     \rho
     \\ 0 < \size{G} }}
(-1)^{\size{G}-1}\size[\Big]{\Delta_{B, G}}\\
&= \sum_{\substack{ G \subseteq A
    %\setminus     \rho
    \\ 0 < \size{G} \le b}}
(-1)^{\size{G}-1} \; \size*{\Delta_{B,G} }.
}
The last equality holds, because by assumption $\forall
G\subset V \;\; \size{G} > b
\Rightarrow \Delta_{B, G}=\emptyset$.
\end{proof}

% With a preprocessing, we remove any constraint $\bigwedge_{v \in G} v
%   \Rightarrow \quad \bigvee_{v \in H}  v$ such that $G \cap H \ne
%   \emptyset$. Moreover, for all constraints $\lnot
% \bigwedge_{v \in G} v$ in $C$, we remove any other constraints $\bigwedge_{v \in G} v
%   \Rightarrow \quad \bigvee_{v \in H}  v$. This does not change the
%   set of satisfying assignments of the instance.

The machine $M$ of instance $J$ can \emph{lookup} the values 
$\size{\Delta_{B, G}}$ and $\Lambda_{B, G}$ 
for $\size{B} \le k_0$ and $\size{G}\le b+1$.
\noindent
$M$ starts with a blank tape and operates in $3$ steps and accepts if no rejection occurs.

\noindent
\textit{Step 1.
%  :  Guess an assignment of size $k$.
} $M$ chooses nondeterministically the set of variables
$A\subseteq V$ of size $ \size{A}\le k_0$.\\
%If $k \le b$ then ...\\
%\\FIXME $M$ rejects if $\size{\Delta_{A, \emptyset}} > 0$ and accepts otherwise.
\textit{Step 2.} $M$ iterates over all $B, G \subseteq A$,
such that $\size{G} \le b+1$, and rejects if $\Lambda_{B, G}> b$.\\
\textit{Step 3.
  % : Counting by the principle of inclusion-exclusion.
} $M$ iterates over all $B \subseteq A$,  and 
%such that $\size{\rho} \le h$,
performs the summation on the right side of Equation (\ref{inclusion-exclusion})
 and rejects if the result is not equal to $\size{\Delta_{B, \emptyset}}$.

 We show now that if instance $I$ has a
 satisfying assignment $E$, then the machine $M$ accepts.  Assignment $E$ satisfies every
 constraint in $C$, and so has size $\size{E} = k_0$. Consider the
 computation branch of $M$, in which  in Step $1$ the set $A$ is chosen equal to
 $E$. In this branch, $M$  rejects neither in Step $2$, nor in  Step $3$ (by
 Fact \ref{claim2} and Claim \ref{claim1}), therefore $M$ accepts.

For the other direction, we prove that if machine $M$ accepts, then  in every
accepting branch of computation, $A$ is a satisfying assignment of $I$. %The hypothesis of \ref{claim1} holds for all $B$ that
%are the premises of constraints
%that have $\le h$ variables in their conclusion. This in turn implies, by
%definition of $\Theta^h$ and preprocessing, that
We have to only consider those constraints $CW_{d, m}^{\br{b}}\ang{\omega}$ in $C$ such that
$\omega\p{\br{d}}\subseteq A$. As there is no rejection in Step 2, we have
$\size{\omega^{-1} \p{A} \cap \br{d+1, d+m}}
\le  b$. So the premises
of Claim \ref{claim1} hold  for
 sets $A$ and $\omega\p{\br{d}}$, and thus Equation (\ref{inclusion-exclusion}) holds. Therefore,
by Fact \ref{claim2}, as there is no rejection in Step $3$, we have
$\size{\omega \p{\br{d+1, d+m}} \cap A} \ge 1$. So  $A$ satisfies the
constraint.

\noindent\textbf {More Details}

For each variable $v \in V$, there is a symbol $\sigma_v$ in the alphabet of $M$, with
an alphabetical order over these symbols. Sets
of variables are presented by the alphabetically sorted string of their symbols.

To do the summation in Equation  (\ref{inclusion-exclusion}) in Step 3,
each time the calculated partial sum is added to the next summand. There are $
\le \binom{k_0}{b}$ summands, and each summand is an
integer between $-n$ and $n$. Therefore the partial sum is always between $ -n
\binom{k_0}{b+1}$ and $ n \binom{k_0}{b+1}$. Machine $M$ has a symbol $s_i$ for each integer
$i$ in this range and \textit{knows} how to add two such integers. 

By definition,  $ \size{\Delta_{B,G}} < \size{C}$ and $\Lambda_{B, G}\le n$.
These values are
\textit{stored} in $M$, only if $\size{G} \le$ $\min \p{ b+1, k_0}$ and $\size{B}
\le k_0$, and there is a  constraint $CW_{d, m}^{\br{b}}\ang{\omega}$ in $C$, such that $\omega\p{\br{d}} =B$ and
$G \subseteq \omega \p{\br{d+1, d+m}}$.
%with premises  $ \bigwedge_{v \in B} v$ in $C$.
Therefore the number of stored values is
% $ \le n \sum_{i=0}^{b+1}\binom{\lambda}{i}$.
$ \le n \sum_{i=0}^{b+1}\binom{n}{i}$.
The
lookups for other values will safely be answered by value $0$. These values
are stored  and accessed by using the \textit{trie} data structure, with the pairs $\p{B,G}$ as
the key, and the  proper $s_i$ symbol as the value. Notice that the length of
the key is bounded by a fixed function of $k_0$.  The above two arguments show
that the runtime of the reduction is bounded by
% $O\p{ n^c   \sum_{i=0}^{b+1}\binom{\lambda}{i}}$
$O\p{\sum_{i=0}^{b+1}\binom{n}{i}} n^{O\p{1}}$. %$O\p{\binom{\lambda}{b }n^c}$.

% By choosing the set $P$
% In Step 1,
% $M$ is \textit{programmed} to consider only
% those sets $V_i$ that $k_i > 0$, and there are $\le p$ such sets.
In steps
2 and 3, $M$ iterates over subsets of $A$, all of size $\le k_0$. As discussed above, the time needed for each lookup and summation is
bounded by a 
function of $k_0$, therefore  the runtime of $M$ is bounded by $f\p{k_0}$ for a fixed function $f$.
\end{proof}
%{\color{blue} $\p{W}_{t, k_0}$}
\begin{lemma}
  \label{combine}
Let $b \ge 0$. Then  \csp$ \p{W}_{k,t} \p{CW^{\br{b}}}_{k}$ is in \W$\br{1}$.
\end{lemma}
\begin{proof}
We reduce the problem to \textsc{Short NonDeterministic Turing Machine
  Acceptance}. Let $I\defeq\p{V, C}$ be the given instance. Let $I_1\defeq\p{V, C_1}$ and
$I_2\defeq\p{V, C_2}$ be the corresponding instances of
\csp$ \p{W}_{k,t}$ and \csp$ \p{CW^{\br{b}}}_{k}$,  such that $C= C_1 \cup C_2$
and the parameter values of $I_1$ and $I_2$ are that of $I$. Notice that $V$
is the set of variables of $I_1$, $I_2$ and  $I$.

  By Theorem \ref{appearance},  there is a  nondeterministic Turing machine
  $M_1$ such that a set  $A_1\subseteq V$ is a satisfying assignment of
$I_1$ if and only if a branch of $M_1$
selects the set $A_1$ in Step $1$ and accepts.

By Theorem \ref{reduction},  there is a  nondeterministic Turing machine $M_2$
such that a set $A_2 \subseteq V$ is a satisfying assignment of $I_2$ if and
only if a branch of $M_2$
selects the set $A_2$ in Step $1$ and accepts.

Consider the  the nondeterminist Turing machine $M$ described in the following. $M$ nondeterministically simulates $M_1$
and then nondeterministically simulates $M_2$. A branch of $M$ accepts if and
only if $M_1$ and $M_2$ accept on the  corresponding branches and $A_1 = A_2$.

If $M$ accepts, then clearly $A_1$($=A_2$) is a satisfying assignment of $I$. 
On the other hand, if a set $A\subseteq V$ is a satisfying assignment of $I$,
then $A$ is a satisfying assignment of $I_1$ and $I_2$  by definition, therefor $M$ has an
accepting branch of computation where $A_1=A_2=A$.
\end{proof}

\section{Partial sets and their Completions}

\begin{definition}
  \label{completion-partial-def}
Given a relation $R$ (of arity say     $q$), and a set  $T\subseteq \br{q}$, $T \not
\in R$,   a  \emph{completion} of $T$ is a minimal set $U$ such that
\eq[set]{
  &U \in R\\
 & T \subset U. %\\
%  &\nexists W\in R \quad T\subset W \subset U.
}
We denote the set of all completions of $T$ by $\completion_R\p{T}$.

We say that a set $T_1\subset \br{q}$ is \emph{partial} (to $R$) if $T_1 \not\in
R$ and any partial $T_2 \subset T_1$ has a completion that is a subset of $T_1$. We
denote the set of all sets partial to $R$ by $\partl\p{R}$.
\end{definition}
Notice that if $\emptyset \not \in R$, then $\emptyset \in \partl\p{R}$.

\begin{fact}
If $T$ is a minimal set such that $T \subseteq \br{q}$ and $T \not \in R$, then
$T$ is partial.
\end{fact}
\begin{fact}\label{partial-has-completion} 
  Let  $W \in R$. If $T \subset W$  
  and $T \not \in R$,
  then there is a set $U \subseteq W$ such that $U \in \completion_R\p{T}$.
\end{fact}

It is not hard to see that for a  set $D \subseteq \br{q}$, we have $D\in R$ if and only if for every $T \not
\in R$ such that $T \subseteq D$, there is $U \in \completion_R\p{T}$ such that
$U \subseteq D$. But we can be
much more efficient: $D\in R$ if and only if for every $T \in \partl\p{R}$ such that $T \subseteq D$, there is $U \in \completion_R\p{T}$ such that
$U \subseteq D$. The reduction in the proof of the following theorem is based
on this idea.

\begin{theorem}\label{completion-size-lemma} 
  Let $\Gamma=\cb{R_i}_{i\in \mathbb{N}}$ be a (possibly infinite)
  constraint language and  there be an integer
$d \ge 1$ such that for any $R \in \Gamma$, for every $T \in R$, it holds $\size{T}  \le d$.
  Then \csp$\p{\Gamma}_{k} $ is in \W$\br{1}$.
\end{theorem}
\begin{proof}
  We show that   \csp$\p{\Gamma}_{k}$ is fpt-reducible  to
  $  \csp\p[\big]{W}_{t, k_{\le}}\p[\big]{CW^{\br{2^d}}}_{k_{\le}}$.
  The result follows by Lemma \ref{combine} and Lemma \ref{k-le-to-k}.
  Given an instance $I \defeq \p{V_1, C}$ with parameter $k\p{I} = k_0$, 
  we  construct an instance $J\defeq\p{V_1\cup V_2, C'}$ as follows.

  First notice that because $\Gamma$ is fpt membership checkable, for every $i
  \in \mathbb{N}$, the set $\partl\p{R_i}$ can be
computed in time $O\p{r^d} \p{\log i}^{O\p{1}}$, where $r\defeq \arity\p{R_i}$. Also for each $T \in \partl\p{R}$ the set
$\completion_R\p{T}$ can be computed in time $O\p{r^d} \p{\log i}^{O\p{1}}$.

  For each constraint $C\p{i}\defeq R\ang{\omega_1}$, $i \in \br{2, \size{C}}$, and for each $T$ in $\partl\p{R}$, 
  introduce the new variable
$\lambda_{\omega_1\p{T}}$ and the set of new variables $B$,
\eq{
  B\defeq \cb{\lambda_{\omega_1\p{U}} | U \in \completion_R\p{T}  }
  }
(if they are not already introduced), and add a constraint $CW_{1, \size{B}}^{\br{2^d}}\ang{\omega_2}$ to $C'$, where $\omega_2\p{1}
  =\lambda_{\omega_1\p{T}}$ and $\omega_2\p{\br{2,
      \size{B}+1}} =B$.
  
 For each new variable $\lambda_E$ introduced above ($E\subseteq V_1$, a set of variables), add a constraint
  $CW_{\size{E},1}^{\br{1}}\ang{\omega}$ to $C'$ where $\omega\p{\br{\size{E}}}=E$ and $\omega\p{\size{E}+1}=\lambda_E$.
  Also, for every variable $x \in E$, add the constraint
  $CW_{1, 1}^{\br{1}}\ang{\omega}$ to $C'$, where $\omega\p{1}=\lambda_E$ and $\omega\p{2}= x$.  
  We call these \emph{binding} constraints (notice that for every $b\ge 0$, it holds
   $CW_{b,1}^{\br{1}} \in CW^{\br{2^d}}$). 

% With a trivial post-processing, we assure that each variable appears in each of
% the added constraints at most once. 
  
Set $V_2$ to be the set of all newly introduced variables.
Add to $C'$ the constraint $W_{\size{V_1}}^{\cb{k_0}}\ang{\omega_3}$  where $   \omega_3\p{\br{\size{V_1}}}\defeq V_1$.
      Finally, add the constraint $C'\p{1} \defeq   W_{\size{V_1}+\size{V_2}}^{\br{k_0 + 2^{k_0}}}\ang{\omega_4}$,  where 
   $\omega_4\p{\br{\size{V_1}+\size{V_2}}}\defeq V_1 \cup V_2$.

 This sets parameter
$k_\le$ of $J$  to $k_\le\p{J} = k_0+2^{k_0}$.

\begin{fact}\label{binding}
 For two sets $D\subseteq V_1$ and $Q\subseteq V_2$, let $D \cup Q$ satisfy $J$. Then the binding
constraints in $C'$ ensure that
\eq{
\forall \lambda_E \in V_2 \quad\quad \lambda_E \in Q \;\text{ iff }\;
E\subseteq  D.
}

\end{fact}
By Fact \ref{binding}, it is enough to prove the following claim to show that $I$ has a satisfying assignment if and only if $J$ has a
satisfying assignment. 
\begin{claim}\label{D_and_Q}
  Let sets $D\subseteq V_1$ and $Q\subseteq V_2$ be such that %$D \cap
%  V_{p+1} = \emptyset$ 
   $Q=\cb{\lambda_E|\lambda_E \in
  V_2,\,E \subseteq D}$. Then the assignment $D$
satisfies $I$ iff $D \cup Q$ 
satisfies $J$.
\end{claim}
\begin{proof}
 If $D$ is a  solution of  $I$, let $\psi$ be a constraint in $C'$. If $\psi$ is a binding constraint, then it is trivially satisfied by
$D\cup Q$. Otherwise, let $\psi$ be a constraint $CW^{\br{2^d}}_{1,
    \size{B}}\ang{\omega_2}$ that corresponds to some constraint $R\ang{\omega_1}
\in C$ and a set $T \in \partl(R)$ (with set $B$ as defined in the
construction above).
So $\omega_2\p{1} = \lambda_{\omega_1\p{T}}$ and $\omega_2\p{\br{2,
      \size{B}+1}} =B$. Now if $\omega_1\p{T} \not \subseteq D$, then $\lambda_{\omega_1\p{T}} \not\in Q$ and 
assignment $D$ trivially   satisfies $\psi$. Else if $\omega_1\p{T} \subseteq D$,
first, $\lambda_{\omega_1\p{T}} \in Q$.
Secondly, we apply the Fact \ref{partial-has-completion}:
We have $\omega_1^{-1}\p{D} \in R$,  $T \subset \omega_1^{-1} \p{D}$ and $T\not \in R$,
therefore there is a $U \in \completion_R\p{T}$ such that $U \subseteq
\omega_1^{-1} \p{D}$, which means $\omega_1\p{U} \subseteq D$ (notice that the
number of such $U$ is at most $2^d$).
Therefore  variable $\lambda_{\omega_1\p{U}}$ in
set $V_{2}$ exists and $\lambda_{\omega_1\p{U}} \in Q$.
It follows that 
$\psi$ is satisfied by assignment
$D\cup Q$. 

For the other direction, assume for the sake of contradiction  that $D \cup Q$ 
satisfies $J$,  but $D$ does not satisfy $I$. The first constraint $C\p{1}$ is
trivially satisfied, so let 
 $R\ang{\omega}$ be another constraint of $C$  such
that $ \omega^{-1}\p{D} \not \in R$. 
By Definition \ref{completion-partial-def},  there is $T \in \partl\p{R}$ such that $T \subseteq
\omega^{-1}\p{D}$ and for every $U \in \completion_R\p{T}$, it holds that $U \not \subseteq$ $\omega^{-1}\p{D}$.
There is, by construction, variable
 $\lambda_{\omega\p{T}} \in V_2$, and  $\lambda_{\omega\p{T}}\in Q$. 
And either $\completion_R\p{T} = \emptyset$ or for every $U\in
\completion_R\p{T} $, 
variable $\lambda_{\omega\p{U}} \not \in Q$.
Therefore, the constraint $\psi$ in $J$ that (by construction) corresponds to
$R\ang{\omega}$ and $T$, is not satisfied by assignment $D \cup Q$, a contradiction.
\end{proof}
\end{proof}
% \begin{corollary}\label{main-strong}
% If a (possibly infinite) constraint language $\Gamma$ is fpt enumerable, and there be a function $g$ such that
%     $M_q^h \le \p{\log q}^{g\p{h}}$,
% then \csp$\p{\Gamma}_{k, e}$ $\in$ \W$\br{1}$.
%     \end{corollary}
%     \begin{proof}      
% By Lemma \ref{n-power-log-k-fpt}, $M_n^{e_0\cdot k} \le $ is an \fpt\ function,
% therefore by Lemma \ref{completion-size-lemma},  \csp$\p{S}_{t, k}\p{\Gamma}_e$
% $\le_{\text{Turing}}^{\fpt}$ \csp$\p{S}_{t, k}\p{\cb{H_{d,m,m}}_{0 \le d, m}}_{l,
%   e=1}$. The result follows by Theorem \ref{in-w1-l}.      
% \end{proof}

  \begin{corollary} \label{in-w1}
    % (i)
   Let $d \ge 0$. Then   \csp$ \p{W^d}_{k} \in  \W\br{1}$.
    % (ii) For the c.l. $\Gamma$,
    % \eq{
    %   \Gamma \defeq \cb{CW_{d,m}^{A} | 0 \le d, m; A\subseteq\br{0,m}}
    % }
    % \csp$\p{S}_{t, k}\p{\Gamma}_{l\p{m}, \max_A}$ is \W$\br{1}$-complete.    
%    \csp$\p{S}_{t, k}\p{\cb{CW_{d,m}^{\br{\alpha}}}_{0 \le d, m, \alpha}}_{l\p{m}, \alpha_{\max}}$ is \W$\br{1}$-complete.    
\end{corollary}
% \subparagraph*{Acknowledgements.}

% I want to thank \dots

% \appendix

% \section{Morbi eros magna}
% a

%%
%% Bibliography
%%

%% Either use bibtex (recommended), 

\bibliography{kuk}

%% .. or use the thebibliography environment explicitely

 \end{document}